\documentclass[copyright]{eptcs}
\providecommand{\event}{QAPL 2016} 

\usepackage[T1]{fontenc}
\usepackage{xcolor}
\hypersetup{
    colorlinks,
    linkcolor={red!50!black},
    citecolor={blue!50!black},
    urlcolor={blue!80!black}
}
\usepackage{graphicx}
\usepackage{booktabs}
\usepackage{verbatim}
\usepackage[linesnumbered, ruled, onelanguage]{algorithm2e}
\usepackage{wrapfig}
\usepackage{amsmath, amssymb, amsthm}
\usepackage{amsfonts}
\usepackage{enumerate} 
\usepackage[]{todonotes} 
\usepackage{xspace}
\usepackage{longtable}
\usepackage{siunitx}
\usepackage{makecell}

\makeatletter
\@ifpackagelater{siunitx}{2011/06/15}
{%
  \sisetup{ table-comparator = true }
  
}
{%
  
}%
\makeatother

\title{Sound Probabilistic \ssat with Projection}
\author{
\begin{tabular}{c@{\hspace{3em}}c@{\hspace{3em}}c}
  Vladimir Klebanov & Alexander Weigl & J\"org Weisbarth\\
    \texttt{\footnotesize klebanov@kit.edu} &
    \texttt{\footnotesize weigl@kit.edu} &
    \texttt{\footnotesize joerg.bretten@web.de}
\end{tabular}\\[-2ex]
  \institute{Institute for Theoretical Informatics\\
    Karlsruhe Institute of Technology\\
    Germany}
}
\def\titlerunning{Sound Probabilistic \protect\ssat with
Projection}
\def\authorrunning{V.~Klebanov, A.~Weigl \& J.~Weisbarth}

\newcommand{\dfn}[1]{\emph{#1}}

\renewcommand{\vec}[1]{\ensuremath{\bar{#1}}}

\newcommand{\Z}{\ensuremath{\mathbb Z}}
\newcommand{\N}{\ensuremath{\mathbb N}}
\newcommand{\Fml}{\ensuremath{\text{Fml}}} 
\newcommand{\Md}[2]{\ensuremath{\mathop{models}_{#1}({#2})}} 
\newcommand{\nsimeq}{\not\simeq}
\newcommand{\pivot}{\ensuremath{\mathit{pivot}}}
\renewcommand{\phi}{\varphi}
\newcommand\restr[2]{{
  \left.\kern-\nulldelimiterspace 
  #1 
  \vphantom{\big|} 
  \right|_{#2} 
  }}
\newcommand{\ssat}{%
{\settoheight{\dimen0}{S}\resizebox{!}{\dimen0}{\raisebox{\depth}{\#}}SAT}\xspace}
\newcommand{\bssat}{\textsc{Bounded\#SAT}}
\newcommand{\core}{\textsc{Core}}
\newcommand{\main}{\textsc{Main}}
\newcommand{\origmc}{\textsc{ApproxMC}\xspace}
\newcommand{\ourmc}{\textsc{ApproxMC-p}\xspace}

\newcommand{\model}{\textsc{SAT}}

\newcommand{\vocab}[1]{\ensuremath{\mathit{vocab}(#1)}}
\newcommand{\phidelta}{\ensuremath{\restr{\phi}{\Delta}}}

\newcommand{\syntrue}{\ensuremath{\mathit{true}}}

\newcommand{\semtrue}{\ensuremath{\textsc{True}}}
\newcommand{\semfalse}{\ensuremath{\textsc{False}}}

\newcommand{\dsharp}{\textsc{Dsharp}\xspace}
\newcommand{\sharpsat}{\textsc{sharpSAT}\xspace}
\newcommand{\sharpcdcl}{\textsc{sharp\-CDCL}\xspace}

\newcommand{\cryptominisativ}{\textsc{Cryptomini\-sat4}\xspace}
\newcommand{\cryptominisatv}{\textsc{Cryptomini\-sat5}\xspace}

\newtheorem{theorem}{Theorem}[section]
\newtheorem{definition}[theorem]{Definition}
\newtheorem{construction}[theorem]{Construction}
\newtheorem{lemma}[theorem]{Lemma}
\newtheorem{corollary}[theorem]{Corollary}

\newtheorem{note}[theorem]{Note}

\newcommand{\ED}{\ensuremath{(\epsilon,\delta)}\xspace}

\begin{document}
\maketitle

\begin{abstract}
  We present an improved method for a sound probabilistic estimation of
  the model count of a boolean formula under projection. The problem
  solved can be used to encode a variety of quantitative program
  analyses, such as concerning security of resource consumption. We
  implement the technique and discuss its application to quantifying
  information flow in programs.
\end{abstract}

\section{Introduction}

The \ssat problem is concerned with counting the number of models of a
boolean formula. Since \ssat is a computationally difficult problem, not
only exact but also approximative solutions are of interest. A powerful
approximation method is probabilistic  approximation,  making use of
random sampling. We call a probabilistic approximation method
\emph{sound} when the probability and magnitude of the sampling-related
error can be bounded a~priory. A prominent recent development in sound
probabilistic \ssat is \origmc~\cite{ChakrabortyMV13}.

In this paper we present \ourmc, an improved method and a tool for sound
probabilistic \ssat with projection. Just as \origmc, which it
enhances, \ourmc belongs to the category of \ED~counters. In this
context, the parameter~$\epsilon$ represents the \emph{tolerance} and
$1-\delta$ the confidence of the result. For example, choosing
$\epsilon=0.1$ and $\delta=0.14$ implies that the computed result
provably lies with a probability of~86\% between the $0.9$-fold and
the~$1.1$-fold of the correct result. Both parameters can be configured
by the user.

An enticing application of \ssat solvers is quantitative program
analysis, which often requires establishing cardinality of sets defined
in terms of the set of reachable program states. Reducing an analysis to
\ssat has the advantage that one can use a variety of established
reasoning techniques. At the same time, it is reasonably easy to
represent behavior of intricate low-level programs in boolean logic.
Yet, \ssat alone is typically not sufficient for this purpose---one
needs a way to reason about program reachability. In logic, this
reasoning corresponds to projection. If a boolean formula encodes a
relation (e.g., a transition relation on states), then  computing the
image or the preimage of the relation is a projection operation.

\ourmc takes as input a boolean formula in conjunctive normal form
together with a projection scope (a set of variables)  and the
parameters~$\epsilon$~and~$\delta$ and estimates the number of models of
the formula projected on the given scope. Of course, by setting the
scope to encompass all variables in the logical signature, one can use
\ourmc as a non-projecting \ssat solver.

%

The specific contributions of this paper are the following:

First, we materially improve the performance of \origmc, in particular
its base confidence. Probabilistic counters meet confidence
demands above the base confidence by repeating the estimation. We reduce
the number of repetitions for confidence values above~0.6 by about an
order of magnitude on average. Furthermore, we reduce, for one
repetition,  the number of SAT solver queries by at least~$15\%$. A
detailed comparison is presented in Section~\ref{sec:modif}.

Second, we combine probabilistic model counting with projection, even
though the ideas behind this combination are not
completely new. A particular special case has previously appeared
in~\cite{Chakraborty14} in the context of uniform model sampling. There,
the formula~$\phi$ is treated by considering only its projection
on the \emph{independent support}. An independent support  is a subset
of variables that uniquely determines the truth value of the whole
formula. It is often known from the application domain. For formulas
generated from deterministic programs, for instance, the independent
support is the preimage of the transition relation. Our approach is
more general in that we explicitly consider projection on arbitrary
scopes.

Third, we implement the method and map its pragmatics. We show that
\ourmc is effective for large formulas with a large number of models,
which may make other \ssat tools run out of time or memory. Finally, we
discuss applications of \ourmc to quantifying information flow in
programs.

\subsection{Logical Foundations}

We assume that logical formulas are built from usual logical connectives
($\wedge$, $\vee$, $\neg$, etc.) and  propositional variables from some
vocabulary set~$\Sigma$.   A \dfn{model} is a map assigning every
variable in~$\Sigma$ a truth value. A given model~$M$ can be
homomorphically extended to give a truth value to a formula~$\phi$
according to standard rules for logical connectives. We call a model~$M$
a \dfn{model of}~$\phi$, if~$M$ assigns~$\phi$ the value \semtrue.  A
formula~$\phi$ is \dfn{satisfiable} if it has at least one model, and
\dfn{unsatisfiable} otherwise.

In the following, we assume that $\Sigma$ and $\Delta$ are vocabularies
with $\Delta\subseteq\Sigma$. A $\Sigma$-entity (i.e., formula or model)
is an entity defined (only) over vocabulary from~$\Sigma$. We assume
that $\phi$ denotes a $\Sigma$-formula and $M$~a $\Sigma$-model.
With $\vocab{\phi}$ we denote the vocabulary actually appearing in~$\phi$.

With $\Md{\Sigma}{\phi}$ we denote the set of all models
of~$\phi$. If $\phi$ is unsatisfiable, the result is the empty
set~$\emptyset$.
With $|\phi|$ we denote the number of models of the formula~$\phi$
(i.e., $|\phi|=|\Md{\Sigma}{\phi}|$).
With $\restr{M}{\Delta}$ we denote the $\Delta$-model that coincides with the
$\Sigma$-model~$M$ on the vocabulary~$\Delta$.

With $\restr{\phi}{\Delta}$ we denote the projection of~$\phi$
on~$\Delta$, i.e., the strongest $\Delta$-formula that, when interpreted
as a $\Sigma$-formula, is entailed by~$\phi$. The projected
formula~$\restr{\phi}{\Delta}$ says the same things about~$\Delta$
as~$\phi$ does---but nothing else.  Projection of~$\phi$ on~$\Delta$ can
be seen as quantifying the $\Sigma\setminus\Delta$-variables in~$\phi$
existentially and then eliminating the quantifier (i.e., computing an
equivalent formula without it). Furthermore,
$\Md{\Delta}{\restr{\phi}{\Delta}}=\{\restr{M}{\Delta}\mid
M\in\Md{\Sigma}{\phi}\}$.

\subsection{Related Work}


A number of \textbf{exact boolean model counters} exist. Counters such as
\dsharp~\cite{Muise12} and \sharpsat~\cite{Thurley06} are based on
compiling the formula to the Deterministic Decomposable Negation Normal
Form (d-DNNF). They are geared toward formulas with a large number of
models but tend to run out of memory as formula size increases. An
extension of the above with projection has been presented
in~\cite{KlebanovMM13}.
Another class of exact counters implements variations of the blocking clause
approach (cf.~Section~\ref{sec:ssat}) and includes tools such as
\sharpcdcl~\cite{KlebanovMM13} and CLASP~\cite{gebser2007clasp}. The
counters in this class are often already projection-capable. They can
deal with very large formulas (hundreds of megabyte in DIMACS format)
but are challenged by large model counts.


The \textbf{probabilistic counters} can be divided into three classes.
The first class are the already mentioned \ED~counters. These counters
  were originally introduced by Karp and Luby~\cite{karp1989monte} to
  count the models of DNF formulas. They guarantee with a probability of
  at least~$1-\delta$ that the result will be between $1-\epsilon$ and
  $1+\epsilon$ times the actual number of models. An instantiation of this
  class for CNF formulas is \origmc~\cite{ChakrabortyMV13}.

  For reasons unknown to us, \origmc\ deviates from the original
  definition~\cite{karp1989monte} of an \ED~counter by defining the
  tolerable result interval as $\left[|\phi|/(1 + \epsilon),
  |\phi|\cdot(1 + \epsilon)\right]$. We {adhere to the original
  definition} of the tolerable interval, that is $\left[|\phi|\cdot(1 -
  \epsilon), |\phi|\cdot(1 + \epsilon)\right]$. We discuss the
  differences between \origmc and our work in detail in Section~\ref{sec:modif}.

The second class are lower/upper bounding counters. These counters drop
  the tolerance guarantee and compute an upper/lower bound for the
  number of models that is correct with a probability of at least
  $1-\delta$ (for a user-specified~$\delta$). Examples are
  \textsc{BPCount}~\cite{kroc2008leveraging},
  \textsc{MiniCount}~\cite{kroc2008leveraging},
  \textsc{MBound}~\cite{Gomes2006} and
  \textsc{Hybrid-MBound}~\cite{Gomes2006}.

The third class are guarantee-less counters. These counters provide no formal
  guarantees but can be very good in practice. Examples are
  \textsc{ApproxCount}~\cite{wei2005new},
  \textsc{SearchTreeSampler}~\cite{ermon2012uniform},
  \textsc{SE}~\cite{rubinstein2013stochastic} and
  \textsc{SampleSearch}~\cite{gogate2011samplesearch}.


\section{Method}
\label{sec:approxmc}

\subsection{The Idea}
\label{sec:approxmc-functionality}

The intuitive idea behind \ourmc\ is to partition the set
$\Md{\Sigma}{\phidelta}$ into buckets so that each bucket contains
roughly the same number of models. The partitioning is based on strongly
universal hashing and is, surprisingly, attainable with high probability
without any knowledge about the structure of the set of models. The
count of $\Md{\Sigma}{\phidelta}$ can then be estimated as the count of
models in one bucket multiplied with the number of buckets.

More technically, \ourmc\ is based on Chernoff-Hoeffding bounds, one of
the so called \emph{concentration inequalities}. This theorem
(Theorem~\ref{theorem:1}) limits the probability that the sum of random
variables -- under certain side conditions -- deviates from its expected
value. To apply the theorem we make use of a trick common in counting,
namely that set cardinality can be expressed as the sum of the
membership indicator function over the domain.

Assume that we fixed the set of buckets~$B$,
a way of distributing models of~$\phi$ into buckets by
means of a hash function~$h$, and distinguished one particular bucket.
We now associate each model  with an indicator variable (a random
variable over the hash function~$h$) that is $1$ iff the model is within
the distinguished bucket and $0$ otherwise. Hence, for a given
hash function~$h$, the sum of all those indicator variables is exactly
the amount of models within the distinguished bucket. We will determine
this value by means of a deterministic model counting procedure \bssat\
(Section~\ref{sec:ssat}).

On the other hand, the \emph{expected} number of models in the bucket
when choosing $h$ randomly from the class of  strongly $r$-universal
hash functions (Section~\ref{sec:runiv}) is
${|\Md{\Sigma}{\phidelta}|}/{B}$. The Chernoff-Hoeffding theorem tells
us that the measured and the expected values are probably close and
allows us to estimate~$|\Md{\Sigma}{\phidelta}|$.

We first explain how to build an $(\epsilon,\delta)$ counter this way
for a fixed confidence $1-\delta\approx 0.86$ (Section~\ref{sec:fixed}).
Then, we generalize this result to arbitrary higher confidences
(Section~\ref{sec:amplif}).


The idea behind adding projection capability is to hash partial  models
(i.e., models restricted to the projection scope) and to use a
projection-capable version of \bssat. To separate concerns, we postpone
discussing projection until Section~\ref{sec:proj}. The algorithms we
present in the following are capable of projection, but we begin with a
tacit assumption that they are always invoked with the value of the
scope~$\Delta=\Sigma$, i.e., in a non-projecting fashion. We will show
that this assumption is superfluous in Section~\ref{sec:proj}.

\subsection{Strongly \texorpdfstring{$r$}{r}-Universal (aka
\texorpdfstring{$r$}{r}-Wise Independent) Hash Functions}
\label{sec:runiv}

The key to distributing models of a formula into a number of buckets
filled roughly equally is to apply strongly $r$-universal
hashing~\cite{Wegman81} (also known as $r$-wise independent or, simply,
$r$-independent hashing). Every concrete strongly $r$-universal hash
function depends on a parameter. By choosing the parameter at random,
one can make a good distribution of values into hash buckets likely,
even when the keys are under adversarial control.

\begin{definition}[strongly \texorpdfstring{$r$}{r}-universal hash functions~\cite{Wegman81}]
Let~$K$ be some universe from which the keys to be hashed are drawn, and
$B$~a set of buckets (hash values). A family of hash functions $H=\{h\colon K\to
B\}$ is \dfn{strongly $r$-universal}, iff for any $h\in H$ chosen
uniformly at random, the hash values of any $r$-tuple of distinct keys
$(k_1,\ldots,k_r)\in K^r$ are independent random variables, i.e., for
any $r$-tuple of (not necessarily distinct) values $(v_1,\ldots,v_r)\in
B^r$
\[
\Pr_{h\in H}\big[h(k_1)=v_1\wedge\ldots\wedge h(k_r)=v_r\big]=
\left(\frac{1}{|B|}\right)^r\enspace.
\]
\end{definition}

In the following, we are interested in families of strongly
$r$-universal hash functions with $K=\Z_2^n$ and $B=\Z_2^m$. We denote
any such family as~$H(n,m,r)$. Functions~$h\in H(n, m, r)$ can be used
to distribute models with~$n$ variables into~$2^m$ buckets. While we
keep the rest of the presentation generic,  our implementation resorts
to a particular family $H_\mathit{xor}(n, m, 3)$ of such functions
with $r=3$. Any discussion of concrete values refers to this family and
the corresponding degree of strong universality. Note that the
construction operates on $\Z_2\cong\{0,1\}$, and that addition on~$\Z_2$
corresponds to exclusive or (boolean XOR), while multiplication
on~$\Z_2$ corresponds to conjunction (boolean AND). Otherwise, the usual
matrix and vector arithmetic rules apply.

\begin{construction}[Hash function family $H_\mathit{xor}(n, m, 3)$]
Let $n$ and $m$ be arbitrary natural numbers.
Any $m\cdot(n + 1)$ values $b_{1;0}, \ldots, b_{m;n} \in \Z_2$
define a \emph{hash function} $h\colon\Z_2^n\to\Z_2^m$ by
%
%
\[
\vec{z} \mapsto
\begin{pmatrix}
  b_{1;0} \\
  \vdots \\
  b_{m; 0}
\end{pmatrix}
\oplus
\begin{pmatrix}
  b_{1; 1} & \cdots & b_{1; n} \\
  \vdots & \ddots & \vdots \\
  b_{m; 1} & \cdots & b_{m; n}
\end{pmatrix}
\otimes
\begin{pmatrix}
  z_1 \\
  \vdots \\
  z_n
\end{pmatrix}\enspace.
\]
We denote the class of all such hash functions as  $H_\mathit{xor}(n, m,
3)$.
\end{construction}
\begin{theorem}[\cite{GomesSS06}]
The hash function class $H_{\mathit{xor}}(n, m, 3)$ is
strongly $3$-universal.
\end{theorem}

By sampling $b_{1;0}, \ldots, b_{m;n}$ uniformly from~$\Z_2$, we can
sample uniformly from~$H_{\mathit{xor}}(n, m, 3)$.

\begin{construction}\label{cons:select}
Let $h\in H_{\mathit{xor}}(n, m, 3)$ be a hash function. Fixing its
output induces a predicate on its inputs. In this paper, we will
consider for each~$h$, the predicate
\(
h(\vec{z})=1^m
\).
This semantical predicate can be represented syntactically as a formula
of propositional logic built from $m$~XOR clauses:
\begin{equation}\label{eq:select}
\big(b_{1;0}\oplus (b_{1; 1}\wedge z_1)\oplus\ldots\oplus (b_{1; n}\wedge
z_n)\big)
\wedge\ldots\wedge
\big(b_{m;0}\oplus (b_{m; 1}\wedge z_1)\oplus\ldots\oplus (b_{m; n}\wedge
z_n)\big)\enspace.
\end{equation}
\end{construction}

\begin{construction}\label{cons:phih}
Given a hash function~$h$, we use the notation~$\phi_h$ to denote the
conjunction of a formula~$\phi$ with the clause
representation~\eqref{eq:select} of the predicate induced by~$h$.
\end{construction}

For presentation purposes, we assume w.l.o.g.\ that
$\Sigma=\{z_1,\ldots,z_n\}$ and $\Delta=\{z_1,\ldots,z_{|\Delta|}\}$ in
the following. Our implementation does not have this limitation.

\pagebreak[4]
\subsection{Helper Algorithm \bssat: Iterative Model Enumeration}
\label{sec:ssat}

\begin{wrapfigure}{r}{0.43\textwidth}
\vspace{0pt}
\begin{algorithm}[H]
  \caption{\rlap{\bssat($ \phi, \Delta, n $)}}
  \label{alg:projmodenum}
  \SetAlgoLined

  $ k \leftarrow 0 $ \;
  $ M \leftarrow \model(\phi) $ \;
  \While{$ (M \neq \bot)\wedge(k<n) $}{
    $ k \leftarrow k +1 $ \;
    $ \phi \leftarrow \phi \wedge (\Delta \nsimeq \restr{M}{\Delta}) $ \;
    $ M \leftarrow \model(\phi) $ \;
  }
  \Return{$k$}
\end{algorithm}
\end{wrapfigure}

To enumerate the models in a single bucket we are using the well-known
algorithm \bssat\ (Algorithm~\ref{alg:projmodenum}). Given a
formula~$\phi$, a projection scope $\Delta \subseteq \Sigma$, and a
bound $n\geqslant0$, the algorithm  enumerates up to~$n$ models of
$\restr{\phi}{\Delta}$, i.e., it returns
$\mathop{min}(|\restr{\phi}{\Delta}|, n)$. The algorithm  makes use of
the oracle $\model(\cdot)$, which for a CNF formula returns either a
model  or~$\bot$, in case none exists. \bssat\ works by repeatedly
asking the oracle for a model~$M$ of~$\phi$, and extending~$\phi$ with a
\emph{blocking clause} $(\Delta \nsimeq \restr{M}{\Delta})$ ensuring
that any model found later must differ in at least one
$\Delta$-variable.  The formula $\Delta\not\simeq\restr{M}{\Delta}$ can
be constructed as $\bigvee_{v\in\Delta}\mathit{flip}(v,M)$, where
$\mathit{flip}(v,M)=v$, if $M(v)=\semfalse$, and
$\mathit{flip}(v,M)=\neg v$, if $M(v)=\semtrue$. The algorithm is
widely-known as part of the automated deduction lore. We have reported
on our experiences with using it for quantitative information flow analysis
in~\cite{KlebanovMM13}.


\subsection{Counting with a Fixed Confidence of 86\%}
\label{sec:fixed}

\begin{algorithm}
  \label{alg:approxmccore}
  \caption{\core($r$, $\phi$, $\Delta$, $\epsilon$)}
  \SetAlgoLined

  $ n := |\Delta| $ \;
  $ \pivot := \left\lceil \frac{2 \cdot r \cdot (1 + \epsilon)
      \cdot \sqrt[3]{e}}{\epsilon^2} \right\rceil $ \;
  \label{alg:approxmc_pivot}
  $ m \leftarrow 0 $ \;
  \label{alg:approxmccore_m_init}
  \Repeat{$ c \leq \pivot\label{l:repii}
    \; \vee \; m > \left\lceil \log_2 \left(
        \frac{(1 + \epsilon)\cdot 2^n}{\pivot} \right) \right\rceil$\label{exitcond}}{\label{l:repi}
    $ m \leftarrow m + 1 $ \;

    $ h \overset{\text{random}}{\leftarrow} H(n, m, r) $ \;
    $ c \leftarrow $ \bssat($ \phi_h,\Delta, \pivot + 1 $)
    \tcp*{\smash{$c = \min(|\restr{\phi_h}{\Delta}|, \pivot + 1) $}}
    \label{alg:approxmccore_c_assignment}
  }
  \label{alg:approxmccore_loopcond}

  \Return{$c \cdot 2^m$}
  \tcp*{Count of one bucket times number of buckets}
  \label{alg:approxmccore_return}
\end{algorithm}

We first explain how to build an $(\epsilon,\delta)$ counter for a fixed
confidence $1-\delta\approx 0.86$: the algorithm \core. The number stems
from the probability of deviation $e^{\lfloor -r / 2 \rfloor}\approx
0.14$ for $r=3$ in the following theorem.

\begin{theorem}[Chernoff-Hoeffding bounds with limited
independence~\cite{schmidt1995chernoff}]
  \label{theorem:1}
  If $ \Gamma $ is the sum of $r$-wise independent random variables, each of which is confined to the
  interval $ [0, 1] $, then the following holds for $ \mu := E[\Gamma] $ and for every $ \epsilon
  \in [0, 1] $: If $ r \leq \lfloor \epsilon^2 \cdot \mu / \sqrt[3]{e} \rfloor $ then $ \Pr\left[
    |\Gamma - \mu| \geq \epsilon \cdot \mu \right] \leq e^{\lfloor -r / 2 \rfloor} $.
\end{theorem}

\begin{corollary}
  \label{corollary:1}
    If $ r \leq \lfloor \epsilon^2 \cdot \mu / \sqrt[3]{e} \rfloor $, then
    $ \Pr\left[(1 - \epsilon) \cdot \mu \leq \Gamma \leq
              (1 + \epsilon) \cdot \mu \right] \geq
      1 - e^{\lfloor -r / 2 \rfloor} $.
\end{corollary}

\begin{construction}
Let $h$ be chosen randomly from $H(n, m, r)$. For each $M\in
\Md{\Sigma}{\phi}$, we define an integer random variable~$\gamma_M$
(random over~$h$) such that
\[
\gamma_M :=
\begin{cases}
1,\qquad\text{if } h(M)=1^m\\
0,\qquad\text{otherwise.}
\end{cases}
\]
%
The sum of these random variables we denote by  $\Gamma := \sum_{M \in
\Md{\Sigma}{\phi}} \gamma_M$.
\end{construction}
Clearly, $\Gamma = |\phi_h|$ (cf.~Construction~\ref{cons:phih}).
\begin{lemma}\label{thm:r-univ}
For the above construction, the following holds:
\begin{enumerate}
\item The variables~$\gamma_M$ are $r$-wise independent.
\item For any $M\in \Md{\Sigma}{\phi}$, $\Pr[\gamma_M=1]=1/2^m$.
\end{enumerate}
\end{lemma}
\noindent The expectation of~$\Gamma$ (over~$h$) is thus
$\mu := E[\Gamma] = \sum_{M \in
\Md{\Sigma}{\phi}} E[\gamma_M] = \sum_{M \in \Md{\Sigma}{\phi}}
2^{-m} \cdot 1 = |\phi| \cdot 2^{-m}$.
%
Substituting these values into Corollary~\ref{corollary:1}, we obtain:
\begin{lemma}[Models in a hash bucket]
  \label{lemma:effect_of_H}
  Let $ \phi \in \Fml_\Sigma $, $ n := |\Sigma| $, $ \epsilon \in [0,
  1] $, let $ m\in\N$ with $m\leq \lfloor \log_2(|\phi| \cdot
  \epsilon^2 / (r \cdot \sqrt[3]{e})) \rfloor $, $h\in H(n, m, r) $ a
  randomly chosen strongly $r$-universal hash function. It holds:
  \[
  \Pr\left[(1 - \epsilon) \cdot \frac{|\phi|}{2^m}
  \leq |\phi_h|
  \leq (1 + \epsilon) \cdot \frac{|\phi|}{2^m} \right]
  \geq 1 - e^{\lfloor -r / 2 \rfloor}\enspace.
  \]
\end{lemma}

One could think that this lemma is sufficient for estimating~$|\phi|$ by
determining $|\phi_h|$ for some~$m$ (e.g., with \bssat), but,
unfortunately, the upper bound on the admissible values of~$m$ depends
on~$|\phi|$, the very value we are trying to estimate.  This fact forces
us to search for a ``good'' value of~$m$ in \core\
(Lines~\ref{l:repi}--\ref{l:repii}). The search proceeds
in ascending order of~$m$ for reasons of soundness, which will be
explained in the main theorem below. At the same time, smaller values
of~$m$ correspond to larger values of~$|\phi_h|$, which may be
infeasible to count (for $m=0$, for instance, $\phi_h=\phi$). We thus
introduce the counting upper bound \pivot, which only depends on the
tolerance and is defined in \core. The search terminates successfully
when $|\phi_h|\leq\pivot$. We show that this search criterion does not
reduce the probability of correct estimation
(Theorem~\ref{theorem:guarantee}).

\begin{lemma}
\core\ terminates for all inputs.
\end{lemma}
\begin{proof}
The loop has two exit conditions combined in a disjunction
(Line~\ref{exitcond}). Since $m$ is monotonically increasing, the second
exit condition guarantees termination. Note that the second exit
condition makes use of the fact that $|\phi|\leq 2^n$ and is
essentially an emergency stop. It does not, in general, entail that the
algorithm returns a tolerable estimation. This is not a problem, as we
will show that the first exit condition (which implies an estimation
within desired tolerance) terminates the loop sufficiently often for the
desired confidence level.
\end{proof}

\begin{theorem}[Main result]
  \label{theorem:guarantee}
  For $ \epsilon \in (0, 1] $ algorithm \core\ returns with a
  probability of at least $ 1 - e^{\lfloor - r / 2 \rfloor} $ a value
  within $ \left[(1 - \epsilon) \cdot |\phi|, (1 + \epsilon) \cdot
  |\phi| \right] $.
\end{theorem}
\begin{proof}
It is easy to see that $c = \min(|\phi_h|, \pivot + 1)$ is an invariant
of the loop in \core. If the exit condition $c\leq\pivot$ comes
to hold, the invariant dictates that \core\ returns
$2^m\cdot|\phi_h|$.

We now show that there is at least one iteration of the loop (indexed
by $m=m'$) such that  with a probability of at least $1 - e^{\lfloor - r
/ 2 \rfloor}$ the following is true: the exit condition $c\leq\pivot$
holds  and the return value $2^{m'}\cdot|\phi_h|\in\left[(1 - \epsilon)
\cdot |\phi|, (1 + \epsilon) \cdot |\phi| \right]$. But first, we
interrupt the proof for a lemma.

  \begin{lemma}\label{lemma:ex}
  For the given choice of $\pivot$, there exists~$m'$ such that:
  \begin{gather}
  \lceil \log_2((1 + \epsilon) \cdot |\phi| / \pivot) \rceil \leq
  m'\label{ex1}\\
  m' \leq
  \lfloor \log_2(|\phi| \cdot \epsilon^2 / (r \cdot \sqrt[3]{e}))
  \rfloor\label{ex2}\enspace.
  \end{gather}
  \end{lemma}
  \begin{proof}
  It is straightforward to show that
  $\lceil \log_2((1 + \epsilon) \cdot |\phi| / \pivot) \rceil \leq
  \lfloor \log_2(|\phi| \cdot \epsilon^2 / (r \cdot \sqrt[3]{e}))
  \rfloor$.
  \end{proof}

Condition~\eqref{ex2} on~$m'$ fulfills the precondition of
Lemma~\ref{lemma:effect_of_H} and thus entails
  \begin{equation}\label{eq:guar}
  \Pr\left[(1 - \epsilon) \cdot \frac{|\phi|}{2^{m'}}
    \leq |\phi_h|
    \leq (1 + \epsilon) \cdot \frac{|\phi|}{2^{m'}} \right]
  \geq 1 - e^{\lfloor -r / 2 \rfloor}
  \end{equation}
which together with condition~\eqref{ex1}, which is equivalent to
$(1 + \epsilon) \cdot \frac{|\phi|}{2^{m'}} \leq \pivot$, gives
  \begin{equation}
    \label{eq:theorem:guarantee_goal}
    \Pr\left[(1 - \epsilon) \cdot \frac{|\phi|}{2^{m'}}
      \leq |\phi_h|
      \leq (1 + \epsilon) \cdot \frac{|\phi|}{2^{m'}}
      \leq \pivot \right]
    \geq 1 - e^{\lfloor - r / 2 \rfloor}
  \end{equation}
resp.\@
  \[
  \Pr\left[(1 - \epsilon) \cdot \frac{|\phi|}{2^{m'}}
  \leq |\phi_h|
  \leq (1 + \epsilon) \cdot \frac{|\phi|}{2^{m'}}
  \; \wedge \; |\phi_h| \leq \pivot \right]
  \geq 1 - e^{\lfloor - r / 2 \rfloor}\enspace.
  \]
Since we are incrementing~$m$ during search, the last equation implies
that both loop termination and result quality are likely at some point.
We also note that an earlier termination with $m<m'$ is not problematic,
since result quality hinges on condition~\eqref{ex2}, which is an upper
bound on~$m$.
\end{proof}


\subsection{Scaling to Arbitrary Confidence}\label{sec:amplif}

\newcommand{\plaincore}{\bgroup\let\textsc\relax\core\egroup}

\begin{algorithm}
  \caption{\main($r$, $\phi$, $\Delta$, $\epsilon$, $\delta$)}
  \label{alg:approxmc}
  \SetAlgoLined

  $ t := \min\left(\left\{ n \in \mathbb{N}:
  \delta \geq
  \sum_{k = \lceil n / 2 \rceil}^{n} \binom{n}{k}
  \cdot e^{\lfloor -r / 2 \rfloor \cdot k} \cdot (1 - e^{\lfloor -r / 2 \rfloor})^{n - k}
  \right\}\right) $ \;
  \label{alg:approxmc_t}
  $ \pivot := \left\lceil \frac{2 \cdot r \cdot (1 + \epsilon)
      \cdot \sqrt[3]{e}}{\epsilon^2} \right\rceil $ \;

  $ c \leftarrow $ \bssat($ \phi, \Delta, \pivot + 1 $) \;
  \label{alg_approxmc_boundedsat}
  \eIf{$ c \leq \pivot $}{
    \Return{$ c $} \tcp*{solution is exact with confidence $1$}
    \label{alg:approxmc_ret_exact}
  }{
    $\mathit{samples}\leftarrow $ repeat($t$, \core($r$, $\phi$, $\Delta$, $\epsilon$))
    \label{alg:approxmc_repeat}
    \tcp*{do \plaincore\ $t$ times}

    \Return{\mbox{\rm median}($\mathit{samples}$)} \;
  }
\end{algorithm}

Since \core\ returns the correct result with a probability of
approx.~$0.86>0.5$, it is possible to amplify the confidence by
repeating the experiment. This is what algorithm \main\ does. To prove
its correctness  (Theorem~\ref{theorem:3}), a technical lemma is needed
first.
\begin{lemma}[Biased coin tosses and related estimations]
  \label{lemma:coin_tosses}
  Let $p := \Pr[\text{head}]$ be the probability of tossing head with a biased coin.
  \begin{enumerate}[(a)]
  \item \label{lemma:coin_tosses_1} The probability to toss $ m $ times
    head in $n\geqslant m$ independent coin tosses is
    \( \binom{n}{m} \cdot p^m \cdot (1 - p)^{n - m} \)
  \item \label{lemma:coin_tosses_2} The probability of tossing at least
    $m$ heads is:
    \( \sum_{k = m}^n \binom{n}{k} \cdot p^k \cdot (1 - p)^{n - k} \)

  \item \label{lemma:coin_tosses_3} (Geometric series) For every non-negative
    integer $n$:
    \( \sum_{k = 0}^n x^k = \frac{1 - x^{n + 1}}{1 - x} \)
  \item \label{lemma:coin_tosses_4} The probability to toss at least
    $m = \lceil n / 2 \rceil$ times head for $p\in[0, 1/2]$ can
    be bounded from above:
    \[
    \sum_{k = \lceil n / 2 \rceil}^n \binom{n}{k} \cdot p^k \cdot (1 - p)^{n - k}
    \leq \frac{1 - p}{1 - 2 \cdot p}
    \cdot \left( \sqrt{4 \cdot p \cdot (1 - p)} \right)^n\enspace.
    \]
  \end{enumerate}
\end{lemma}

\begin{theorem}[Theorem 3 in~\cite{ChakrabortyMV13}]
  \label{theorem:3}
  Let $\phi$ be a formula,  $\delta$ and $\epsilon$ parameters in
  $(0, 1]$, and $\tilde{c}$~an output of $\main(r, \phi, \epsilon,
  \delta)$. Then
  $\Pr\big[(1-\epsilon)\cdot|\phi|\leqslant\tilde{c}\leqslant
  (1+\epsilon)\cdot|\phi|\big]\geqslant 1-\delta$.
\end{theorem}
\begin{proof}
  If $|\phi|\leq\pivot$, \main\ returns the exact solution
  (Algorithm~\ref{alg:approxmc}, Line~\ref{alg:approxmc_ret_exact}). If
  not, the algorithm returns the median~$\tilde{c}$ of
  $t$~probabilistic
  estimations ($t$~is defined in Algorithm~\ref{alg:approxmc},
  Line~\ref{alg:approxmc_t}).
  The goal is to show that the probability of~$\tilde{c}$
  being outside the
  tolerance is at most~$\delta$.

  A necessity for~$\tilde{c}$ being outside the tolerance is that at
  least $\lceil t / 2 \rceil$ of the estimations of \core\ are
  outside the tolerance, due to the definition of the median. The
  probability that a single estimation of \core\ is outside the
  tolerance is at most $e^{\lfloor -r / 2 \rfloor}$ by
  Theorem~\ref{theorem:guarantee}. Now, the probability to have at
  least $\lceil t / 2 \rceil$ estimations (out of~$t$
  estimations in total) outside the tolerance can be seen as the
  probability to toss at least $\lceil t / 2 \rceil$ times head in a
  series of $t$~coin tosses where $\Pr[\text{head}]=e^{\lfloor -r / 2
  \rfloor}$. By Lemma~\ref{lemma:coin_tosses}~\ref{lemma:coin_tosses_2},
  this probability is:
  \[
  \sum_{k = \lceil t / 2 \rceil}^t \binom{t}{k}
  \cdot e^{\lfloor -r / 2 \rfloor \cdot k}
  \cdot \left(1 - e^{\lfloor -r / 2 \rfloor} \right)^{t - k}\enspace.
  \]
  Due to the choice of~$t$, this probability is smaller than~$\delta$.
  The existence of~$t$ is ensured, because $t$ can be bounded from
  above per Lemma~\ref{lemma:coin_tosses}~\ref{lemma:coin_tosses_4}:
  \[
  t \leq \max\left(1, \left\lceil
  \log_{\sqrt{4 \cdot e^{\lfloor -r / 2 \rfloor} \cdot (1 - e^{\lfloor -r / 2 \rfloor})}} \left(
    \delta \cdot \frac{1 - 2 \cdot e^{\lfloor -r / 2 \rfloor}}{1 - e^{\lfloor -r / 2 \rfloor}}
  \right) \right\rceil \right)\enspace.
  \]
  The value of $t$ can be (pre-)computed by a simple search (see
  Table~\ref{tab:t}).
%
%
\end{proof}

\begin{note}[Leap-frogging]\label{sec:approxmc:enhancements:loopcycles}
We observe that every repetition of \core\ begins the search for the
proper number of buckets with $m = 1$. It is natural to ask if the
repeated search can be abridged. In~\cite{ChakrabortyMV13}, a heuristic
called \emph{leap-frogging} is proposed. Leap-frogging tracks the
successful values of~$m$ (i.e., the ones upon termination) as \core\ is
repeated. After a short stabilization period, subsequent runs of \core\
begin the search not with $m=1$ but with the minimum of the successful
values observed so far. The authors report that leap-frogging is
successful in practice. We choose to abstain from it nonetheless, as
leap-frogging nullifies all soundness guarantees.

A sound optimization is possible if one knows a lower bound~$L$ on the
number of models of~$\phi$ resp.~$\restr{\phi}{\Delta}$. In this case,
it is sound to start the search with $m=\lceil \log_2((1 + \epsilon)
\cdot L / \pivot) \rceil$ as per proof of
Theorem~\ref{theorem:guarantee}.
\end{note}

\subsection{Counting with Projection}
\label{sec:proj}

To show that \ourmc works properly for $\Delta\subset\Sigma$, we
show that the result in this case is the same as computing~$\phidelta$
in some other way and then applying \ourmc as a non-projecting
counter (i.e., as discussed so far). We begin with a lemma.

\begin{lemma}\label{lemma:projswap}
If $h$ only contains vocabulary from~$\Delta$ (and constants), then
$(\restr{\phi}{\Delta})_h\equiv\restr{\varphi_h}{\Delta}$.
\end{lemma}
\begin{proof}
First, since $h$ only contains vocabulary from~$\Delta$,
$\restr{h}{\Delta}=h$. Second, projection distributes over conjunction
(elementary). Together:
$(\restr{\phi}{\Delta})_h\equiv
\restr{\phi}{\Delta}\wedge h\equiv
\restr{\phi}{\Delta}\wedge \restr{h}{\Delta}\equiv
\restr{(\varphi\wedge h)}{\Delta}\equiv
\restr{\varphi_h}{\Delta}$.
\end{proof}

We are now interested to establish result equality in the following two
executions:
\begin{equation}\label{eq:proji}
\core(r,\phi,\Delta,\epsilon)=
\core(r,\restr{\phi}{\Delta},\Delta,\epsilon)\enspace.
\end{equation}
We note that since the first, third, and fourth parameters are identical
in both invocations, it is sound to assume the same random choices
of~$h$ in both executions.  Thus, the validity of~\eqref{eq:proji}
reduces to the following equality (assuming an arbitrary~$h$ that can be
chosen by the algorithm):
\begin{equation}\label{eq:projii}
\bssat(\phi_h,\Delta,n)=
\bssat((\restr{\phi}{\Delta})_h,\Delta,n)
\end{equation}
This equality follows from Lemma~\ref{lemma:projswap}, the functionality
of \bssat, the observation that
$\vocab{\restr{\phi}{\Delta}}\subseteq\Delta$ (by definition of
projection), and $\vocab{h}\subseteq\Delta$ (by third
parameter when invoking \core).

Observe that in the execution on the right, all invocations are
non-projecting, if one considers $\Sigma=\Delta$, and are thus correct
according to the previous proofs.


\section{Implementation and Evaluation}

\subsection{Changes and Improvements over the Original \origmc}
\label{sec:modif}

\begin{wrapfigure}{r}{0.52\textwidth}
\vspace{-6ex}
  \begin{center}
    \includegraphics[scale=0.43]{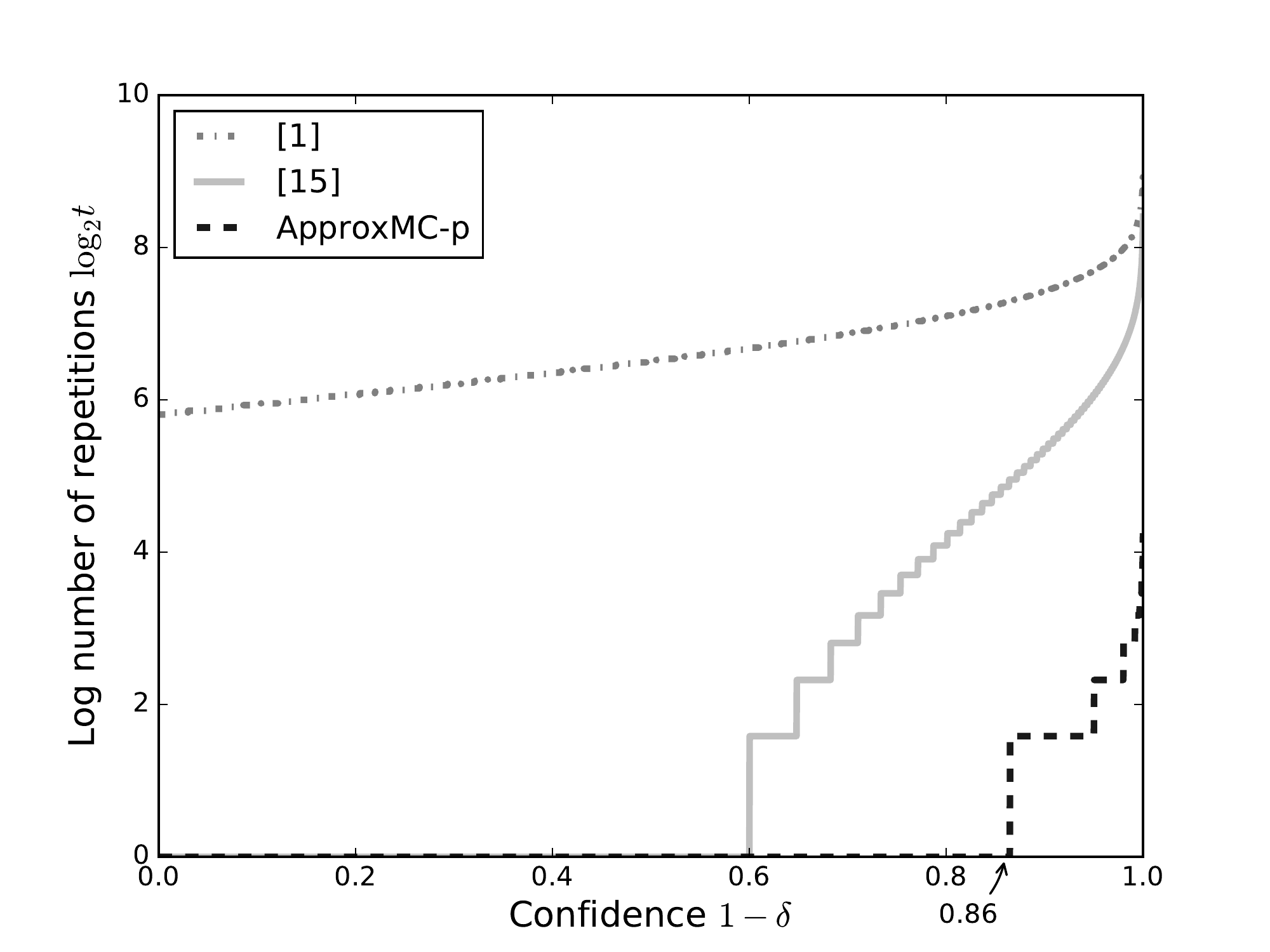}
  \end{center}

  \small{
  The upper curve visualizes $t(1 - \delta) = \lceil 35 \cdot \log_2(3
  / \delta) \rceil$ as used in~\cite{ChakrabortyMV13}. The middle one is $t(1 - \delta) =
  \min\left(\left\{ n \in \mathbb{N}: \delta \geq \sum_{k = \lceil n / 2
  \rceil}^{n} \binom{n}{k} \cdot (2/5)^{k} \cdot (1 - (2/5))^{n - k}
  \right\}\right)$ as used in the master
  thesis~\cite{meel2014sampling} connected to
  the publication~\cite{ChakrabortyMV13}. The lower curve is the
  function in Algorithm~\ref{alg:approxmc}, Line~\ref{alg:approxmc_t}
  for $r=3$.
  }
  \caption{Comparison of required number of repetitions~$t$ of \core\
  (logarithmic scale) for desired level of confidence~$1-\delta$ }
  \label{fig:comp_req_iterations}
\end{wrapfigure}

In comparison to the original \origmc~\cite{ChakrabortyMV13}, following
differences are of note:

\textbf{Elimination of $ \bot $.} The original versions of \core\ and
  \main\ potentially returned an error value $\bot$ instead of
  an estimation. This distinction has not been exploited to achieve
  smaller values for $ \pivot $ and $ t $ and hence has been dropped
  here for more compact and more understandable proofs.

\textbf{Test $|\phi| \leq \pivot$ hoisted out of the inner loop.} The
  result of this test does not change when repeating \core, so it has
  been moved to \main.

\textbf{Smaller values for $\pivot$.} We have performed a more exact
  estimation of the needed value of $\pivot$. Our value of $\pivot$, as
  defined in Algorithm~\ref{alg:approxmc},
  Line~\ref{alg:approxmc_pivot}, is \emph{at least} $15\%$ smaller than
  in~\cite{ChakrabortyMV13} (cf.~Table~\ref{tab:pivot}).


\textbf{Increased base confidence (fewer repetitions).} The number of
  repetitions needed for the demanded confidence has been substantially
  reduced (Table~\ref{tab:t} and Figure~\ref{fig:comp_req_iterations}).
  The reasons for this are twofold. First, in~\cite{ChakrabortyMV13},
  the Chernoff-Hoeffding inequality is missing the floor operator,
  unnecessarily reducing the base confidence. Second,
  in~\cite{ChakrabortyMV13}, the ``successful'' termination of the loop
  in \core\ and the tolerable quality of the achieved estimation are
  treated as independent events and their probabilities multiplied. We
  show in Theorem~\ref{theorem:guarantee} that the latter actually
  entails the former.

\begin{table}
  \caption{Sample values for $\pivot$ depending on the tolerance}\smallskip
  \label{tab:pivot}
  \centering\small
  \begin{tabular}{l*{9}{r}}
    \toprule
    & \multicolumn{9}{c}{tolerance ($ \epsilon $)}\\
    \cmidrule(lr){2-10}
    & 0.75 & 0.5 & 0.25 & 0.1 & 0.05 & 0.03 & 0.01 & 0.005 & 0.001 \\
    \cmidrule(lr){2-2}
    \cmidrule(lr){3-3}
    \cmidrule(lr){4-4}
    \cmidrule(lr){5-5}
    \cmidrule(lr){6-6}
    \cmidrule(lr){7-7}
    \cmidrule(lr){8-8}
    \cmidrule(lr){9-9}
    \cmidrule(lr){10-10}
    $ \pivot $ in~\cite{ChakrabortyMV13} & 54 & 90 & 248 & 1198 & 4364 & 11662
    & 100912 & 399660 & 9912124 \\
    $ \pivot $ now & 27 & 51 & 168 & 922 & 3517 & 9584 & 84575 & 336622 & 8382049\\
    \bottomrule
  \end{tabular}
\end{table}

\begin{table}
  \caption{Sample values for $t$ depending on the confidence}\smallskip
  \label{tab:t}
  \centering\small
  \begin{tabular}{l*{9}{r}}
    \toprule
    & \multicolumn{9}{c}{confidence ($1 - \delta$)}\\
    \cmidrule(lr){2-10}
    & 0.5 & 0.6 & 0.7 & 0.8 & 0.9 & 0.95 & 0.99 & 0.999 & 0.9999 \\
    \cmidrule(lr){2-2}
    \cmidrule(lr){3-3}
    \cmidrule(lr){4-4}
    \cmidrule(lr){5-5}
    \cmidrule(lr){6-6}
    \cmidrule(lr){7-7}
    \cmidrule(lr){8-8}
    \cmidrule(lr){9-9}
    \cmidrule(lr){10-10}
    $ t $ in~\cite{ChakrabortyMV13} & 91 & 102 & 117 & 137 & 172 & 207 & 289 & 405 & 521 \\
    $ t $ in~\cite{meel2014sampling} & 1 & 1 & 7 & 17 & 41 & 67
    & 133 & 235 & 339 \\
    $ t $ now & 1 & 1 & 1 & 1 & 3 & 3 & 7 & 13 & 19\\
    \bottomrule
  \end{tabular}
\end{table}

\subsection{Implementation}\label{sec:impl}

We have implemented \ourmc
 and make our implementation available.\footnote{%
\url{http://formal.iti.kit.edu/ApproxMC-p/}}  Algorithms
\core\ and \main\ are implemented in Python, while different
implementations of \bssat\ can be plugged in. As the source of random
bits we use Python's Mersenne Twister PRNG, seeded with
entropy obtained from \texttt{/dev/urandom} (Linux kernel PRNG seeded
from hardware noise). We have chosen a pseudo-random number generator
instead of a true randomness source (e.g., the HotBits service used
in~\cite{ChakrabortyMV13}) for reasons of reproducibility and ease of
benchmarking. By reusing the random seed it is possible to reproduce
results and also benchmark different implementations of \bssat.

Our principal \bssat\ implementation is based on
\cryptominisatv~\cite{soos2009extending}, which is a recent successor to
\cryptominisativ. \cryptominisatv has built-in support for model
enumeration and efficient XOR-SAT solving by using the Gauss-Jordan
elimination as an inprocessing step\footnote{Please note that
Gauss-Jordan elimination has to be turned on both during compilation and
runtime.}. We modified the tool  to support projection by shortening the
blocking clauses used for model enumeration.

It is also possible to use other model enumerators such as \sharpcdcl or
CLASP~\cite{gebser2007clasp}. However, these tools are not designed with
native support for XORs. As such they cannot directly parse XOR-CNF
inputs, an issue that we work around by translating XORs into CNF using
Tseitin encoding. The tools do not employ Gauss-Jordan elimination
either, which makes them not competitive against \cryptominisatv in our
scenario. In the future, it would be interesting to benchmark these
enumerators combined with Gauss-Jordan elimination as a preprocessing
step, though they would still lack inprocessing in the style of
\cryptominisatv.

\subsection{Experiments}


An off-the-shelf program verification system together with projection
and counting components is sufficient to implement an analysis
quantifying information flow in programs (QIF)~\cite{Klebanov14}. The
measured number of reachable final program states corresponds to the
min-entropy leakage resp.\ min-capacity of the program~\cite{Smith15}
and can be used as a measure of confidentiality or integrity. For
generating verification conditions we are using the bounded model
checker CBMC~\cite{cbmc}, for projection and counting \ourmc. The
experiments were performed on a machine with an Intel Core i7 860
2.80GHz CPU (same machine as in~\cite{KlebanovMM13}). \ourmc was
configured with the default tolerance $\epsilon=0.5$ and confidence
$1-\delta=0.86$.

\paragraph{Synthetic QIF benchmarks.}

In~\cite{meng2011calculating}, the authors describe a series of QIF
benchmarks, which have become quite popular since then. As was already
noted in~\cite{KlebanovMM13}, the majority of the benchmarks are too
easy in the meantime. We use two scaled-up benchmark instances, which
could not be solved in~\cite{KlebanovMM13} without help of a dedicated SAT
preprocessor~\cite{Manthey12}.

In both benchmarks, the size of the projection scope $|\Delta|=32$, and
$|\phidelta|=2^{32}$. The first benchmark, sum-three-32, contains 639
variables and 1708  clauses. The average run time of \ourmc was
1.2s. The average run time for bin-search-32 (4473 variables,
14011 clauses) was 5.2s. 

\paragraph{Quantifying information flow in PRNGs.}

In \cite{DoerreKlebanov16}, we present an information flow analysis
aimed at detecting a certain class of errors in pseudo-random number
generators (PRNGs). An error is present when the information flow from a
seed of $M$ bytes to an $N$-byte chunk of output is not maximal.
In~\cite{DoerreKlebanov16} we detect such deviations from maximality but
do not obtain a quantitative measure of the flow. While quantifying the
flows needed for practical application in the domain (at least 20~bytes)
is still not feasible, this scenario provides a scalable benchmark.

Here, we are quantifying the flow through the OpenSSL PRNG with
cryptographic primitives replaced by idealizations. For $M=N=10$, the
59-megabyte formula contains 590 thousands variables and 2.5 millions
clauses. \ourmc counts all $2^{80}$ models in 10.5 minutes (631.7s on
average), which is beyond the capabilities of any other counting tool
known to us. The largest flow we could measure in this benchmark was at
15~bytes (measured in a single experiment over the course of 47~hours),
while reaching the count of 14~bytes in the same experiment took only
roughly 32~minutes.

\section{Conclusion}

The experiments show that \ourmc can effectively and efficiently
estimate model counts of projected formulas that are not
amenable to other counting tools. At the same time, \ourmc, like any
tool, has its own particular pragmatical properties, which need to be
carefully considered when choosing a tool for an application.

First, \ourmc offers no approximation advantage for formulas with few
models. For instance, at tolerance level $\epsilon=0.1$, formulas with
fewer than 922 models are counted exactly (cf.~Table~\ref{tab:pivot}).
On the other hand, there is no penalty for these formulas either, as
\ourmc then simply behaves as \bssat, which offers the best pragmatics
for this class of formulas. For formulas with model counts larger but
still comparable with \pivot, \ourmc will perform more SAT queries than
\bssat, due to search for the proper number of buckets and experiment
repetition at confidence values over~$0.86$ (base confidence). We also
note that performance of \ourmc does not increase by lowering the
confidence under~$0.86$.

Concerning enumeration performance, \cryptominisatv is currently the
best overall implementation of \bssat\ due to its built-in support for
XORs. Yet, beside Gauss-Jordan elimination, there are various other
factors influencing performance (if to a lesser degree): non-XOR solver
performance, enumeration algorithm, solver preprocessor, etc. To better
understand the individual contributions of these factors, much more
benchmarking and investigation is needed.

Finally, \ourmc, in general, does not make larger formulas amenable
to counting, merely formulas with more models. For QIF analyses, this
means that \ourmc is attractive for quantifying confidentiality in
systems with large secrets, as acceptable information leakage is coupled
to the secret size. Alternatively, \ourmc can be used for quantifying
integrity and related properties.


\paragraph{Acknowledgment.}

This work was in part supported by the German National Science
Foundation (DFG) under the priority programme 1496 ``Reliably Secure
Software Systems -- RS3.'' The authors would like to thank Mate Soos for
very helpful feedback concerning Cryptominisat.


\bibliographystyle{eptcs}
\bibliography{references}
\appendix
\section{Detailed Proofs}

\begin{proof}[Proof of Corollary~\ref{corollary:1}]
  \begin{align*}
    \Pr[|\Gamma - \mu| \geq \epsilon \cdot \mu] \leq e^{\lfloor -r / 2 \rfloor}
    &\iff P [|\Gamma - \mu| \leq \epsilon \cdot \mu] \geq 1 - e^{\lfloor -r / 2 \rfloor} \\
    &\iff P [\Gamma \in [\mu - \epsilon \cdot \mu, \mu + \epsilon \cdot \mu]]
    \geq 1 - e^{\lfloor -r / 2 \rfloor} \\
    &\iff \Pr[(1 - \epsilon) \cdot \mu \leq \Gamma \leq (1 + \epsilon) \cdot \mu]
    \geq 1 - e^{\lfloor -r / 2 \rfloor}
  \end{align*}
\end{proof}

\begin{proof}[Proof of Lemma~\ref{thm:r-univ}]
~
\begin{enumerate}
\item (Measurable) function application preserves independence.
\item We will prove that strong $r$-universality implies strict
$r-1$-universality. The claim of the theorem corresponds to strict
$1$-universality.

Assuming, there are at least~$r$ distinct keys in~$K$ and
that $B=\{w_1,\ldots,w_{|B|}\}$:
\begin{align*}
\Pr_{h\in H}\big[\bigwedge_{i=1}^{r-1}h(k_i)=v_i\big]=
&&\text{\small existence of $h(k_r)$}\\
\Pr_{h\in H}\big[\bigwedge_{i=1}^{r-1}h(k_i)=v_i\wedge
\bigvee_{j=1}^{|B|}h(k_r)=w_j\big]=
&&\text{\small distributivity}\\
\Pr_{h\in H}\big[\bigvee_{j=1}^{|B|}(\bigwedge_{i=1}^{r-1}h(k_i)=v_i\wedge
h(k_r)=w_j)\big]=
&&\text{\small orthogonal disjuncts}\\
\sum_{j=1}^{|B|}\Pr_{h\in H}\big[\bigwedge_{i=1}^{r-1}h(k_i)=v_i\wedge
h(k_r)=w_j\big]=
&&\text{\small strong $r$-universality}\\
|B|\cdot\left(\frac{1}{|B|}\right)^r=
\left(\frac{1}{|B|}\right)^{r-1}\enspace.
\end{align*}
\end{enumerate}
\end{proof}

%
%
%

\begin{proof}[Proof of Lemma~\ref{lemma:ex}]
  \begin{align*}
    &1 + \log_2 \left((1 + \epsilon) \cdot \frac{|\phi|}{\pivot} \right)
    \leq \log_2 \left(\frac{|\phi| \cdot \epsilon^2}{r \cdot \sqrt[3]{e}} \right) \\
    \iff& \log_2(2) + \log_2 \left((1 + \epsilon) \cdot \frac{|\phi|}{pivot} \right)
    \leq \log_2 \left(\frac{|\phi| \cdot \epsilon^2}{r \cdot \sqrt[3]{e}} \right) \\
    \iff& \log_2 \left(2 \cdot (1 + \epsilon) \cdot \frac{|\phi|}{pivot} \right)
    \leq \log_2 \left(\frac{|\phi| \cdot \epsilon^2}{r \cdot \sqrt[3]{e}} \right) \\
    \iff& 2 \cdot (1 + \epsilon) \cdot \frac{|\phi|}{pivot}
    \leq \frac{|\phi| \cdot \epsilon^2}{r \cdot \sqrt[3]{e}} \\
    \iff& \left\lceil \frac{2 \cdot r \cdot (1 + \epsilon)
        \cdot \sqrt[3]{e}}{\epsilon^2} \right\rceil \leq pivot
  \end{align*}
\end{proof}

\begin{proof}[Proof of Lemma~\ref{lemma:coin_tosses}]
  Lemma~\ref{lemma:coin_tosses}~\ref{lemma:coin_tosses_1} and~\ref{lemma:coin_tosses_3} are widely
  known and~\ref{lemma:coin_tosses_2} is trivial with~\ref{lemma:coin_tosses_1}. Therefore
  only~\ref{lemma:coin_tosses_4} will be proved.

  Firstly, because for all $k\in\{\lceil n / 2 \rceil, \ldots, n \} $ it is true that $
  \binom{n}{k}\leqslant2^n $ it holds:
  \begin{align}
    \label{eq:coin_tosses_4_1}
    \sum_{k = \lceil n / 2 \rceil}^n \binom{n}{k} \cdot p^k \cdot (1 - p)^{n - k}
    &= (1 - p)^n \cdot \left( \sum_{k = \lceil n / 2 \rceil}^n \binom{n}{k}
    \cdot \left( \frac{p}{1 - p} \right)^k \right) \notag \\
    &\leq (1 - p)^n \cdot 2^n
    \cdot \left( \sum_{k = \lceil n / 2 \rceil}^n \cdot \left( \frac{p}{1 - p} \right)^k \right)
  \end{align}

  Secondly, due to the restriction of $ p $ to be in $ [0, 1/2] $ the value $ (p / (1 - p))^{\lceil
    n / 2 \rceil} $ is smaller than or equal to $ (p / (1 - p))^{n / 2} $. Hence it applies:
  \begin{equation}
    \label{eq:coin_tosses_4_2}
    \sum_{k = t}^n \left( \frac{p}{1 - p} \right)^k
    \leq \frac{1 - p}{1 - 2 \cdot p} \cdot \left( \sqrt{\frac{p}{1 - p}} \right)^n
  \end{equation}
  The following estimation shows that:
  \begin{align*}
    \sum_{k = \lceil n / 2 \rceil}^n \left( \frac{p}{1 - p} \right)^k
    &= \sum_{k = 0}^n \left( \frac{p}{1 - p} \right)^k
    - \left( \sum_{k = 0}^{\lceil n / 2 \rceil - 1} \left( \frac{p}{1 - p} \right)^k \right) \\
    &\overset{\text{\ref{lemma:coin_tosses_3}}}{=}
    \frac{1 - \left(\frac{p}{1 - p} \right)^{n + 1}}{1 - \frac{p}{1 - p}}
    - \frac{1 - \left(\frac{p}{1 - p} \right)^{\lceil n / 2 \rceil}}{1 - \frac{p}{1 - p}} \\
    &= \frac{1 - p}{1 - 2 \cdot p}
    \cdot \left( \left( \frac{p}{1 - p} \right)^{\lceil n / 2 \rceil}
    - \left( \frac{p}{1 - p} \right)^{n + 1} \right) \\
    &\leq \frac{1 - p}{1 - 2 \cdot p} \cdot \left( \frac{p}{1 - p} \right)^{\lceil n / 2 \rceil} \\
    &= \frac{1 - p}{1 - 2 \cdot p} \cdot \left( \sqrt{\frac{p}{1 - p}} \right)^n
  \end{align*}

  If Equation~\ref{eq:coin_tosses_4_2} is inserted into Equation~\ref{eq:coin_tosses_4_1} you get:
  \[
  \frac{1 - p}{1 - 2 \cdot p}
  \cdot \underbrace{2^n \cdot (1 - p)^n
    \cdot \left( \sqrt{\frac{p}{1 - p}} \right)^n}_{= \left(\sqrt{4 \cdot p \cdot (1 - p)} \right)^n}
  \]
\end{proof}

%

\end{document}